\pdfoutput=1
\documentclass{toc}

\tocdetails{%
volume=13, number=2, year=2017, firstpage=1,
specissue={\cccBG},  %
received={June 15, 2016},
revised={April 9, 2017},
published={May 15, 2017},
author={Rohit Gurjar, Arpita Korwar, and Nitin Saxena},
plaintextauthor={Rohit Gurjar, Arpita Korwar, Nitin Saxena},
title={Identity Testing for Constant-Width, and Any-Order, Read-Once Oblivious Arithmetic Branching Programs},
runningtitle={Identity Testing for Read-Once Oblivious Branching Programs},
acmclassification={F.2.2},
amsclassification={68Q25, 68W30},
keywords={polynomial identity testing, hitting sets, arithmetic branching programs},
doi={10.4086/toc.2017.v013a002}
}

\newcommand{\A}{\mathbb A}
\newcommand{\F}{\mathbb F}
\newcommand{\Wh}{\cbrace{0, 1, 2, \dots}}
\newcommand{\abs}[1]{\lvert #1 \rvert}

\newcommand{\ceil}[1]{\lceil #1 \rceil}

\newcommand{\transpose}{\mathsf T}
\renewcommand{\C}{\mathcal C}
\newcommand{\M}{M}
\newcommand{\degree}{d} %
\newcommand{\Hit}{\mathcal{H}} %

\newcommand{\lis}[3]{{#1}_1 #2 {#1}_2 #2\dots #2 {#1}_{#3}}
\newcommand{\cbrace}[1]{\left\{#1\right\}}

\newcommand{\comment}[1]{}

\DeclareMathOperator{\Span}{span}
\DeclareMathOperator{\paths}{paths}
\DeclareMathOperator{\supp}{supp} %
\DeclareMathOperator{\Supp}{Supp} %
\DeclareMathOperator{\coeff}{coef} %
\DeclareMathOperator{\rank}{rank}
\DeclareMathOperator{\cha}{char}

\DeclareMathOperator{\diag}{diag}

\newcommand{\x}{\boldsymbol x}

\newcommand{\doublesquare}[1]{[\![#1]\!]}
\renewcommand{\t}{\boldsymbol t}
\newcommand{\s}{\boldsymbol s}
\newcommand{\h}{\boldsymbol h}
\renewcommand{\a}{{\boldsymbol a}}  %
\renewcommand{\b}{{\boldsymbol b}}  %
\renewcommand{\c}{{\boldsymbol c}}  %
\newcommand{\xa}{\x^\a}

\newcommand{\xb}{\x^\b}
\newcommand{\xc}{\x^\c}

\newcommand{\w}{\mathrm{w}} %
\newcommand{\W}{\mathcal{W}} %

\newcommand{\poly}{\mathsf{poly}}
\newcommand{\quasipoly}{\mathsf{quasi\mbox{-}poly}}
\newcommand{\RL}{\cclass{RL}}
\renewcommand{\L}{\cclass{L}}

\newtheorem*{theorem*}{Theorem}

\begin{document}
\begin{frontmatter}

\title{Identity Testing for Constant-Width, and 
Any-Order, %
Read-Once Oblivious Arithmetic Branching Programs\titlefootnote{A \href{http://dx.doi.org/10.4230/LIPIcs.CCC.2016.29}{conference version} of this paper appeared in the 
 Proceedings of the 31st Computational Complexity Conference, 
 2016~\cite{GKS16}.}}

\author[gurjar]{Rohit Gurjar\thanks{Supported by DFG grant TH 472/4.}}
\author[korwar]{Arpita Korwar}
\author[saxena]{Nitin Saxena\thanks{Supported by DST-SERB.}}

\begin{abstract}
We give improved hitting sets for two special cases of Read-once Oblivious Arithmetic Branching Programs (ROABP).
First is the case of an ROABP with known 
order of the variables.
The best previously known hitting set for this case had size $(nw)^{O(\log n)}$ %
where $n$ is the number of variables %
and $w$ is the width of the ROABP\@.
Even for a constant-width ROABP, nothing better than a quasi-polynomial bound
was known.
We improve the hitting-set size for the known-order case
 to $n^{O(\log w)}$.
In particular, this gives the first 
polynomial-size 
hitting set for
constant-width ROABP (known-order).
However, our hitting set only works when the characteristic of the field
is zero or large enough.
To construct the hitting set, we use the concept of the rank of the partial derivative matrix.
Unlike previous approaches
which build up from mapping variables to monomials,
we map variables to polynomials directly.

The second case we consider is that of 
polynomials computable by width-$w$ ROABPs in any 
order of the variables.
The best
previously
known hitting set for this case had size
$d^{O(\log w)}(nw)^{O(\log \log w)}$, where $d$ is the individual degree. %
We improve the hitting-set size to 
$(ndw)^{O(\log \log w)}$.
\end{abstract}

%

\end{frontmatter}

\section{Introduction}
The polynomial identity testing (PIT) problem asks if a given multivariate polynomial
is identically zero. 
The input to the problem is given via an arithmetic model computing a polynomial,
for example, an arithmetic circuit, %
which is the arithmetic analogue of a Boolean circuit. %
The degree of the given polynomial is assumed to be polynomially bounded in the circuit size.
Typically, any such circuit %
can compute a polynomial with exponentially many monomials
(exponential in the circuit size).
Thus, one cannot hope to write down the polynomial in a sum-of-monomials form.
However, given such an input, it is possible to efficiently evaluate the polynomial at a point in the field.
This property enables a randomized polynomial identity test with one-sided error. 
It is known that evaluating a small-degree nonzero polynomial over a random point gives a nonzero
value with a good probability~\cite{DL78,Sch80,Zip79}. 
This gives us a randomized PIT---just evaluate the input polynomial, given as an arithmetic circuit, %
at random points.

Finding an efficient deterministic algorithm for PIT has been a major open question in complexity theory.
The question is also related to arithmetic circuit lower bounds~\cite{Agr05, HS80,KI03}. 
The PIT problem has been studied in two paradigms: (i) blackbox test, where
one can only evaluate the polynomial at chosen points and
(ii) whitebox test, where one has access to the description of the input circuit. %
A blackbox test for a family of polynomials is essentially the same as finding a hitting set---a set of points such that any nonzero polynomial in that family
evaluates to a nonzero value on at least one of the points in the set. 
This work concerns finding hitting sets for a special model called read-once 
oblivious arithmetic branching programs (ROABP).

An \emph{arithmetic branching program} (ABP) is a specialized
arithmetic circuit.  It is the arithmetic analogue of a Boolean
branching program (also known as a binary decision diagram).  It is a
directed layered graph, with edges going from a layer of vertices to
the next layer.  The first and the last layers have one vertex each,
called the source and the sink,
respectively. %
Each edge of the graph has a label, which is a ``simple'' polynomial,
for example, a univariate polynomial.  For any path $p$, its weight is
defined to be the product of labels on all the edges in $p$.  The ABP
computes a polynomial which is the sum of weights of all the paths
from the source to the sink.  Apart from its size, another important
parameter for an ABP is its width.  The width of an ABP is the maximum
number of vertices in any of its layers.  See
\expref{Definition}{def:ABP} for a formal definition of ABP.

ABPs are a strong model for computing polynomials. 
It is known that for any size-$s$ arithmetic circuit 
of degree bounded by $\poly(s)$, one can find an ABP of size
$\quasipoly(s)$ computing the same
polynomial~\cite{VSBR83,Val79,Ber84} (see~\cite{Koi12} for a complete
proof).  Even when the width is restricted to a constant, the ABP
model is quite powerful.  Ben-Or and Cleve~\cite{BOC92} have shown
that width-$3$ ABPs have the same expressive power as
polynomial-sized arithmetic formulas.

An ABP is a \emph{read-once oblivious ABP} or ROABP
if each variable occurs in at most one layer of the edges
and every layer has exactly one variable.\footnote{In a \emph{read-once ABP}, each variable occurs only once
on every source-sink path.
An ROABP is a read-once ABP where every occurrence of a variable is in the same layer.}
The read-once property severely restricts the power of the ABP\@. 
There is an explicit family of polynomials that can be computed by simple depth-$3$ ($\Sigma \Pi \Sigma$)
circuits but requires 
exponential-size 
ROABPs~\cite{KNS16} to compute it.
The order of the variables in the consecutive layers is said to be the \emph{variable order} of the ROABP\@.
The variable order affects the size of the minimal ROABP computing a given polynomial.
There are polynomials which have a small ROABP in one variable order but require
exponential size in another variable order.
Nisan~\cite{Nis91} gave an exact characterization of the polynomials computed by width-$w$ ROABPs
in a certain variable order.
In particular, he gave exponential lower bounds for this model.\footnote{The work of~\cite{Nis91} is actually on non-commutative ABPs but the same results apply to ROABP\@.}

The question of whitebox identity testing of ROABPs has been settled by Raz and Shpilka~\cite{RS05},
who gave a 
polynomial-time
algorithm for this. 
However, though ROABPs are a relatively well-understood model,
we still do not have a 
polynomial-time
blackbox algorithm.
The blackbox PIT question is studied with two variations: one where we know
the variable order of the ROABP and the other where we do not know it.
For known-order ROABPs, Forbes and Shpilka~\cite{FS13} gave the first efficient blackbox test
with $(ndw)^{O(\log n)}$ time complexity, where $n$ is the
number of variables, $w$ is the width of the ROABP and
$d$ is the 
individual-degree %
bound of each variable.
For the unknown-order case, Forbes et al.~\cite{FSS14} gave an 
$n^{O(d \log w \log n)}$-time blackbox test.
Observe that the complexity of their algorithm is quasi-polynomial only when $d$ is small. 
Subsequently, Agrawal et al.~\cite{AGKS15} removed the exponential dependence on the individual degree. 
They gave an $(ndw)^{O(\log n)}$-time blackbox test for the unknown-order case. 
Note that these results remain quasi-polynomial even in the case of constant width.
Studying ROABPs has also led to PIT results for other computational models,
for example, 
subexponential-size 
hitting sets for depth-$3$ multilinear circuits~\cite{OSV15}
and 
subexponential-time
whitebox test for read-$k$ oblivious ABPs~\cite{AFSSV15}.

Another motivation to study ROABPs comes from their
Boolean analogues, called read-once ordered branching programs (ROBP).\footnote{ROBPs
are also known as Ordered Binary Decision Diagrams (OBDDs).}
ROBPs have been studied extensively, with regard to the $\RL$ versus $\L$ question 
(randomized log-space versus log-space).
The problem of finding hitting sets for ROABP can be viewed as an analogue of finding
pseudorandom generators (PRG) for ROBP\@.
A pseudorandom generator for a Boolean function $f$ is
an algorithm which can generate a probability distribution 
(with a small sample space) with the property that $f$ cannot distinguish it from the uniform random distribution
(see~\cite{AB09} for details). 
Constructing an optimal PRG for ROBP, \ie, with $O(\log n)$ seed length or 
polynomial-size sample space,
would imply $\RL = \L$.
Although the known pseudorandom
generators for ROBPs and hitting-set generators for ROABPs 
in similar settings have similar complexity,
there is no known way to translate the construction of one to another.
The best known PRG is of seed length $O(\log^2 n)$ 
($n^{O(\log n)}$-size sample space),
when variable order is known~\cite{N90,INW94, RR99}.
On the other hand, in the unknown-order case, the best known seed length
is of size $n^{1/2 + o(1)}$ \cite{IMZ12}.
Finding an $O(\log n)$-seed PRG even for constant-width known-order ROBPs has been a challenging open question.
Though, some special cases of this question have been solved---width-2 ROBPs~\cite{BDVY13}, 
or nearly solved---permutation and regular ROBPs~\cite{BRRY14,BV10,KNP11,De11,Ste12}. 

Our first result addresses the analogous question in the arithmetic setting. 
We give the first
polynomial-time %
blackbox test for constant-width known-order ROABPs.
However, it works only for zero or large characteristic fields.
Our idea is inspired by the pseudorandom generator 
for ROBPs by Impagliazzo, Nisan and Wigderson~\cite{INW94}.
While their result does not give better PRGs for the constant-width case, we are able to achieve this in
the arithmetic setting.

\begin{theorem*}[\expref{Theorem}{thm:knownROABPhs}]
Let $\mathcal{C}$ be the class of $n$-variate, 
individual-degree-$d$  
polynomials in $\F[\x]$
computed by a width-$w$ ROABP in the variable order $(x_1, x_2, \dots, x_n)$.
Then a hitting set of size $dn^{O(\log w)}$ can be constructed for $\mathcal{C}$,
when $\cha(\F) = 0$ or $\cha(\F) >ndw^{\log n }$.
\end{theorem*}

When $w < n$, the size of our hitting set is smaller than the
previously known hitting sets. Furthermore, even in the regime when $w
\geq n$, the size of our hitting set matches the previously best known
hitting sets.
We show that for a nonzero bivariate polynomial $f(x_1,x_2)$ computed by a width-$w$ ROABP, 
the univariate polynomial $f(t^w, t^w+t^{w-1})$ is nonzero. 
For this, we use the notion of rank of the partial derivative matrix of a polynomial,
defined by Nisan~\cite{Nis91}.
Our argument is that
the rank of the partial derivative matrix of
any bivariate polynomial which becomes zero on $(t^w, t^w+t^{w-1})$
is more than $w$, while
for a polynomial computed by a width-$w$ ROABP, this rank is at most $w$.
We use the map $(x_1,x_2) \mapsto (t^w, t^w+t^{w-1})$ recursively in $\log n$ rounds 
to achieve the above mentioned hitting set.
Our technique has a crucial difference from the previous works on ROABPs~\cite{FSS14,FS13,AGKS15}.
The starting point in all the previous techniques is a monomial map,
\ie, each variable is mapped to a monomial.
On the other hand, we argue with a polynomial map directly (where
each variable is mapped to a univariate polynomial).
We believe that our approach could
lead to a 
polynomial-size
hitting set for ROABPs and we now describe a
concrete construction that we conjecture works.
The goal would be to obtain a univariate $n$-tuple $(p_1(t), \dots, p_n(t))$,
such that any polynomial which becomes zero on $(p_1(t), \dots, p_n(t))$ must have rank or evaluation dimension
higher than $w$.
We conjecture that $(t^r, (t+1)^r, \dots, (t+n-1)^r)$ is one such tuple,
where $r$ is polynomially large (\expref{Conjecture}{con:polyROABP}).

We believe that these ideas from the arithmetic setting can help in
constructing an optimal PRG for constant-width ROBP.

Our second result is for the class of polynomials which are computable by ROABPs in any variable order. %
To be precise, a polynomial $f(\x)$ is in this class if for every permutation of the variables,
there exists an ROABP of width-$w$ that computes $f(\x)$ in that variable order.
This class of polynomials has a slightly better hitting set than the class of polynomials computed by ROABPs 
in a particular variable order, but still no
polynomial-size 
hitting set is known. 
The previously best known hitting set for them has
size $d^{O(\log w)}(nw)^{O(\log \log w)}$ \cite{FSS14}.
We improve this to $(n \degree w)^{O(\log \log w)}$. %

\begin{theorem*}[\expref{Theorem}{thm:commROABPhs}]
For $n$-variate, 
individual-degree-$\degree$  %
polynomials computed by width-$w$ 
ROABPs in any order, %
a hitting set of size $(n \degree w)^{O(\log \log w)}$ can be constructed.
\end{theorem*}

To %
obtain  %
this result 
we follow the approach of Forbes et al.~\cite{FSS14},
which used the notion of rank
concentration or low-support concentration, a technique introduced by Agrawal et al.~\cite{ASS13}.
We achieve rank concentration more efficiently using the basis isolation technique of Agrawal et al.~\cite{AGKS15}.
The same technique also yields a more efficient concentration in depth-$3$ set-multilinear circuits 
(see \expref{Section}{sec:prelim} for the definition).
However, it is not clear if it gives better hitting sets for them.
The best known hitting set for them has size $n^{O(\log n)}$ \cite{ASS13}.

\section{Preliminaries}
\label{sec:prelim}

\subsection{Definitions and 
notation}

We write $[n]$~to denote the set $\{1,2, \dots, n\}$ 
and $[\![d]\!]$ to denote the set $\{0,1, \dots, d\}$. 
The character %
$\x$ will denote a 
list of variables, usually the list
$(x_1, x_2, \dots, x_n)$.
For a field $\F$, we denote 
the ring of polynomials over $\F$ by $\F[\x]$, and 
the field of rational functions over $\F$ by $\F(t)$.
For a set $\x$ of $n$ variables
and for an exponent $\a = (a_1, a_2, \dots, a_n) \in \cbrace{0, 1, 2, \dots}^n$,
we  denote the monomial $\prod_{i=1}^n x_i^{a_i}$ by $\xa$.
The \emph{support} of a monomial $\xa$, denoted by $\Supp(\a)$, 
is the set of \index{support of a monomial, $\Supp(\cdot)$}
variables appearing in that monomial, \ie,
$ \{x_i \mid i \in [n], a_i >{} 0\}$. 
The \emph{support size} of a monomial is the cardinality of its support, 
denoted by $\supp(\a)$. 
A monomial is said to be $\ell$-support if its support size is $\ell$
and $(< \ell)$-support if its support size is $< \ell$.
For a polynomial $P(\x)$, the coefficient of a monomial $\xa$ in $P(\x)$ 
is denoted by $\coeff_P(\xa)$. \index{coefficient of a monomial, $\coeff_P(\cdot)$}

For a monomial $\xa$, 
the sum %
$\sum_i a_i$ is said to be its \emph{degree} and 
$a_i$ is said to be its \emph{degree in variable $x_i$} for each $i$.
Similarly, for a polynomial $P$, its degree (or degree in $x_i$) is the maximum degree 
(or maximum degree in $x_i$) of any monomial in $P$ with a nonzero coefficient. 
We define the \emph{individual degree} of $P$ to be 
 $\max_i\{ \deg_{x_i}(P)\}$,
where $\deg_{x_i}$ denotes degree in $x_i$.

To better understand polynomials computed by ROABPs, we often use 
polynomials over an algebra $\A$, \ie, polynomials whose coefficients 
come from $\A$. 
Matrix algebra is the vector space of matrices equipped with the 
matrix product.
$\F^{m \times n}$ represents the set of all $m \times n$ matrices over the field $\F$.
Note that the algebra of $w \times w$ matrices, has dimension $w^2$.

We often view a vector/matrix with polynomial entries,
as a polynomial with vector/matrix coefficients. \index{polynomial over matrices}
For example,
$$ D(x,y)  = 
\begin{pmatrix} 1+x & y - xy \\ x+y & 1 +xy \end{pmatrix}= 
\begin{pmatrix} 1 & 0 \\ 0 & 1 \end{pmatrix} 1 + 
\begin{pmatrix} 1 & 0 \\ 1 & 0 \end{pmatrix} x + 
\begin{pmatrix} 0 & 1 \\ 1 & 0 \end{pmatrix} y + 
\begin{pmatrix} 0 & -1 \\ 0 & 1 \end{pmatrix} xy\,.
$$
Here, the $\coeff_D$ operator will return a matrix for any monomial,
for example,
\[
  \coeff_D(y)= \begin{pmatrix} 0 & 1 \\ 1 & 0 \end{pmatrix}\,.
\]
For a polynomial $D(\x) \in \A[\x]$ over an algebra, its \index{polynomial over an algebra}
\emph{coefficient space} is the space spanned by its coefficients. \index{coefficient space}

For a matrix~$R$, 
we denote its entry in the $i$-th row and $j$-th column by $R(i,j)$. %

As mentioned earlier, a deterministic blackbox PIT is equivalent to constructing a hitting set.
A set of points $\Hit \in \F^n$ is called a \emph{hitting set}
for a class $\mathcal{C}$ of $n$-variate polynomials if for any nonzero polynomial $P$
in $\mathcal{C}$, there exists a point in $\Hit$ where $P$ evaluates to a nonzero
value. 
\subsection{Arithmetic branching programs}
\label{sec:ABP}
\begin{definition}[Arithmetic Branching Program (ABP)]
\label{def:ABP}
An ABP is a layered directed acyclic graph with $q+1$ layers of vertices
$\{V_0,V_1, \dots, V_{q}\}$ and a source $a$ 
and a sink $b$
such that all the edges of the graph only go
from $a$ to $V_0$, $V_{i-1}$ to $V_i$ for any $i \in [q]$ and
$V_q$ to $b$.
The edges have univariate polynomials as their weights %
and as a convention, the edges going out of $u$
and the edges going into $t$ have constant weights, 
\ie, weights from the field $\F$.
The ABP is said to compute the polynomial
\[
  f(\x) = \sum_{p \in \paths(a, b)} \prod_{e \in p} W(e)\,,
\]

where $W(e)$ is the weight of the edge $e$.
\end{definition}

The ABP has width-$w$ if $\abs{V_i} \leq w$ for all $i \in [\![q]\!]$.
Without loss of generality we can assume $\abs{V_i} = w$ for each $i \in [\![q]\!]$.

It is well-known that the sum over all paths 
in a layered graph can be represented by an iterated matrix multiplication.
To see this, let the set of nodes in $V_i$ be $\{v_{i,j} \mid j \in [w]\}$.
It is easy to see that the polynomial computed by the ABP is the same as
$A^{\transpose} (\prod_{i=1}^{q} D_i ) B $,
where $A, B \in \F^{w \times 1}$ 
and $D_i$ is a $w \times w$ matrix for $1 \leq i \leq q$ such that 
\begin{align*}
A(\ell) &= W(a,v_{0,\ell}) &&\text{ for } 1 \leq \ell \leq w\,,\\
D_i(k, \ell) &= W(v_{i-1,k},v_{i,\ell}) &&\text{ for } 1 \leq \ell,k \leq w \text{ and } 1 \leq i \leq q\,,\\
B(k) &= W(v_{q,k},b) &&\text{ for } 1 \leq k \leq w\,.
\end{align*}

\subsubsection{Read-once oblivious ABP} 
An ABP is called a \emph{read-once oblivious ABP} (ROABP) \index{read-once oblivious ABP (ROABP)}
if the edge weights in different layers 
are univariate polynomials in distinct variables.
Formally,
there is a permutation $\pi$ on the set $[q]$
such that the entries in the $i$th matrix $D_i$ are
univariate polynomials over the variable $x_{\pi(i)}$,
\ie, they come from the polynomial ring
$\F[x_{\pi(i)}]$.
Here, $q$ is the same as $n$, the number of variables.
The order $(x_{\pi(1)}, x_{\pi(2)}, \dots, x_{\pi(n)})$ 
is said to be the variable order of the ROABP\@.

Viewing $D_i({x_{\pi(i)}}) \in \F^{w \times w}[x_{\pi(i)}]$ as a polynomial over the matrix algebra,
we can write the polynomial computed by an ROABP as 
$$f(\x) = A^{\transpose}
 D_1(x_{\pi(1)})  D_2(x_{\pi(2)}) \cdots D_n(x_{\pi(n)}) B\,.$$
An equivalent representation of a width-$w$ ROABP can be 
$$f(\x) = 
 D_1(x_{\pi(1)})  D_2(x_{\pi(2)}) \cdots D_n(x_{\pi(n)})\,,$$
where $D_1 \in \F^{1 \times w}[x_{\pi(1)}]$,
$D_i \in \F^{w \times w}[x_{\pi(i)}]$ for $2 \leq i \leq n-1$ 
and $D_n \in \F^{w \times 1}[x_{\pi(n)}]$.

\subsubsection{Any-order ROABP}    %
A polynomial $f(\x)$ is said to be computed by width-$w$
ROABPs in any order,
if
for every permutation $\sigma$ of the variables,
there exists a width-$w$ ROABP
in the variable order $\sigma$
that computes the polynomial $f(\x)$.
\subsubsection{Set-multilinear circuits}
\label{sec:setmulti}
A depth-$3$ set-multilinear circuit is a circuit of the form
$$f(\x) = \sum_{i=1}^k l_{i,1}(\x_1) \; l_{i,2}(\x_2) \cdots \; l_{i,q}(\x_q)\,,$$
where $l_{i,j}$s are linear polynomials and $\x_1, \x_2, \dots, \x_q$
form of a partition of the 
set $\x$ of variables.
It is known that these circuits are subsumed by ROABPs~\cite{FSS14}. 
However, the polynomials computed by these circuits may not be 
computable by ROABPs in any order. %
For example, the $2n$-variate polynomial $(x_1 + y_1)(x_2 + y_2) \cdots (x_n + y_n)$
has a linear-size set-multilinear circuit.
But, every ROABP in the variable sequence $(\lis{x}{,}{n},\lis{y}{,}{n})$
that computes it has width $\ge 2^n$ (follows from Nisan's characterization~\cite{Nis91}).

Note that
ROABPs can compute polynomials with 
individual-degree  %
$\ge 1$,
but set-multilinear circuits cannot.
It is not known whether all multilinear polynomials computed by ROABPs in any order can also be
computed by 
polynomial-size
set-multilinear circuits.

A set-multilinear circuit has a corresponding polynomial over a commutative algebra.
For the polynomial $f(\x)$ above, consider the polynomial over a $k$-dimensional algebra
$$D(\x) = D_1(\x_1) D_2(\x_2) \cdots D_q(\x_q)\,,$$
where $D_j = (l_{1,j},  l_{2,j}, \dots, l_{k,j})$
and the algebra product is coordinate-wise product.
It is easy to see that $f = (1,1, \dots, 1) \cdot  D$.
Note that the polynomials $D_i$s are over a commutative algebra, 
that is, the order of the $D_i$s in the product does not matter.
Hence, some of our techniques for 
any-order %
ROABPs also work for set-multilinear circuits.

\section{Hitting set for known-order ROABP}

\subsection{Bivariate ROABP} \index{bivariate ROABP}
To construct a hitting set for ROABPs, we start with the bivariate case.
Recall that a bivariate ROABP is of the form
$A^\transpose D_1(x_1) D_2(x_2) B $,
where $A, B \in \F^{w \times 1}$, $D_1 \in \F^{w\times w}[x_1]$ and 
$D_2 \in \F^{w\times w}[x_2]$.
It is easy to see that 
 a bivariate polynomial $f(x_1,x_2)$ computed
by a width-$w$ ROABP can be written as $f(x_1,x_2) = \sum_{r=1}^w g_r(x_1)h_r(x_2)$. 
To construct a hitting set for this polynomial,
we will use the notion of a {partial derivative matrix},
defined by
Nisan~\cite{Nis91} in the context of lower bounds.
Let the individual degree of the polynomial $f \in \F[x_1, x_2]$
be bounded by $\degree$.
The \emph{partial derivative matrix} $M_f$ for $f$ \index{partial derivative matrix}
is a $(\degree+1)\times(\degree+1)$ matrix with 
$$M_f(i,j) = \coeff_f(x_1^i x_2^j) \in \F\,,$$
for all $i,j \in \doublesquare{d}$.
It is known that the rank of the matrix $M_f$
equals the smallest possible width of any ROABP computing $f$ \cite{Nis91}. 

\begin{lemma}[rank $\leq$ width]
\label{lem:rankwidth}
For any polynomial $f(x_1,x_2)= \sum_{r=1}^w g_r(x_1)h_r(x_2)$, 
we have %
$\rank(M_f) \leq w$.
\end{lemma}
\begin{proof}
Let us define $f_r = g_r h_r$, for all $r \in [w]$. 
Clearly, $M_f = \sum_{r=1}^w M_{f_r}$, as $f = \sum_{r=1}^w f_r$.
We will show that $\rank(M_{f_r}) \leq 1$, for all $r \in [w]$. 
As $f_r = g_r(x_1) h_r(x_2)$, its coefficients can be written as a product of 
coefficients from $g_r$ and $h_r$, \ie,
$$\coeff_{f_r}(x_1^i x_2^j) = \coeff_{g_r}(x_1^i) \coeff_{h_r}(x_2^j)\,.$$
Now, it is easy to see that $$M_{f_r} = u_r v_r^{\transpose}\,,$$
where $u_r,v_r \in \F^{d+1}$ with 
$u_r = (\coeff_{g_r}(x_1^i))_{i=0}^d$
and $v_r = (\coeff_{h_r}(x_2^i))_{i=0}^d$.

Thus, $\rank(M_{f_r}) \leq 1$ and $\rank(M_f) \leq w$.
\end{proof}
One can also show that if $\rank(M_f) =w$ then
there exists a width-$w$ ROABP computing $f$.
We skip this proof as we will not need it. 
Now, using the above lemma we give a hitting set
for bivariate ROABPs. 

\begin{lemma}
\label{lem:bivariatehs}
Suppose $\cha(\F) =0$, or $\cha(\F) >d$.
Let
\[
  f(x_1,x_2)= \sum_{r=1}^w g_r(x_1) h_r(x_2)
\]
be a nonzero bivariate polynomial over $\F$
of %
individual degree $d$. 
Then $f(t^w, t^w+t^{w-1}) \neq 0$.
\end{lemma}
\begin{proof}
Let $\tilde{f}(t)$ be the polynomial after the substitution, \ie, $\tilde{f}(t) = f(t^w, t^w+t^{w-1})$.
Any monomial $x_1^ix_2^j$ will be mapped to the polynomial $t^{wi} (t^w+t^{w-1})^j$,
under the mentioned substitution. 
The highest power of $t$ coming from this polynomial is $t^{w(i+j)}$. 
We will cluster together all the monomials for which this highest power is the same,
\ie, $i+j$ is the same. 
The set of coefficients corresponding to any such cluster of monomials will
form a \emph{diagonal} in the matrix $M_f$. The set $\{M_f(i,j) \mid i +j =k\}$
is defined to be the \emph{$k$-th diagonal} of $M_f$, for all $0 \leq k \leq 2d$.
Let $\ell$ be the largest number such that the $\ell$-th diagonal has
at least one nonzero element, \ie,
$$\ell = \max \{i+j \mid M_f(i,j) \neq 0\}\,.$$
As $\rank(M_f) \leq w$ (from \expref{Lemma}{lem:rankwidth}),
we claim that the $\ell$-th diagonal has at most $w$ nonzero elements.
To see this, let $\{(i_1,j_1),(i_2,j_2), \dots, (i_{w'}, j_{w'})\}$
be the set of indices where the $\ell$-th
diagonal of $M_f$ has nonzero elements, \ie, the set $\{(i,j) \mid M_f(i,j) \neq 0, \; i+j = \ell \}$.
Observe that $w' \le d+1$.
As $M_f(i,j) =0$ for any $i+j >{} \ell$, it is easy to see that 
the rows $\{M_f(i_1), M_f(i_2), \dots, M_f(i_{w'})\}$ are linearly independent. 
Thus, $w' \leq \rank(M_f) \leq w$.

Now, we claim that there exists an $r$ with $w(\ell-1) < r \leq w \ell $
such that $\coeff_{\tilde{f}}(t^r) \neq 0$. 
To see this, first observe that the highest power of $t$
to which any monomial $x_1^ix_2^j$ with $i+ j < \ell$ can contribute
is $t^{w(\ell-1)}$.
Thus, for any $w(\ell-1) < r \leq w \ell$, the term $t^r$ can come only from
 the monomials $x_1^ix_2^j$ with $i+ j \geq \ell$.
We can ignore the monomials $x_1^ix_2^j$ with $i+j >{} \ell$ as $\coeff_f(x_1^ix_2^j)=M_f(i,j)=0$, when $i+j >{} \ell$. 
Now, for any $i+j =\ell$, 
the monomial $x_1^i x_2^j$ maps to 
$$t^{w(\ell -j)} (t^w+t^{w-1})^j = t^{w \ell}(1+t^{-1})^j = \sum_{p=0}^j \binom{j}{p} t^{w\ell -p}\,.$$
Hence, for any $0 \leq p < w$,  %
$$\coeff_{\tilde{f}}(t^{w \ell -p}) = \sum_{b=1}^{w'} M_f(i_b, j_b) \binom{j_b}{p}\,.$$ 
Here we assume that if $p >{} j_b$, then $\binom{j_b}{p} =0$.
Writing the above equation in the matrix form, we get,
$$ \begin{pmatrix}\coeff_{\tilde{f}}(t^{w \ell}) \\ \vdots \\ \coeff_{\tilde{f}}(t^{w \ell-w +1}) \end{pmatrix}
=  C \begin{pmatrix} M_f(i_1,j_1) \\ \vdots  \\ M_f(i_{w'}, j_{w'}) \end{pmatrix}\,,
$$
where $C$ is a $w \times w'$ matrix with $C(a,b) = \binom{j_b}{a-1}$,
for all $a \in [w]$ and $b \in [w']$.
If all the columns of $C$ are linearly independent,
then clearly, $\coeff_{\tilde{f}}(t^r) \neq 0$ for some $w(\ell-1) < r \leq w \ell $.
We show the linear independence of the columns in \expref{Claim}{cla:independentRows}.
To show this linear independence we need to assume that the numbers $\{j_b\}_b$
are all distinct. Hence, we need the field characteristic
to be zero or strictly greater than $d$,
as $j_b$ can be as high as $d$ for some $b \in [w']$.  

\begin{claim}
\label{cla:independentRows}
Let $C'$ be the $w' \times w'$ submatrix of $C$ with $C'(a,b) = \binom{j_b}{a-1}$,
for all $a \in [w']$ and $b \in [w']$.
Then $C'$ has full rank. 
\end{claim}
\begin{proof}
We will show that for any nonzero vector $\alpha := (\alpha_1, \alpha_2, \dots, \alpha_{w'}) \in \F^{1\times w'}$, 
$\alpha C'  \neq 0$. 
Consider the polynomial 
$$h(y)= \alpha_1 + \alpha_2 \frac{y}{1!} + \alpha_3 \frac{y(y-1)}{2!} + \cdots + \alpha_{w'} \frac{y(y-1)\cdots (y-w'+2)}{(w'-1)!}\,.$$
As $h(y)$ is a nonzero polynomial 
of
degree bounded by $w'-1$, it can have at most $w'-1$ roots. 
Thus, there exists an $b \in [w']$ such that 
\[
h(j_b) = \sum_{a=1}^{w'} \alpha_a \binom{j_b}{a-1} \neq 0\,.\qedhere
\]
\end{proof}
This concludes the proof of \expref{Lemma}{lem:bivariatehs}.
\end{proof}
As mentioned above, the hitting-set proof works only when
the field characteristic is zero or greater than $d$. 
We given an example over a small characteristic field,
which demonstrates that the problem is not with the proof technique,
but with the hitting set itself.
Let the field characteristic be $2$.
Consider the polynomial $f(x_1,x_2) = x_2^{2} + x_1^{2} + x_1$.
Clearly, $f$ has a width-$2$ ROABP\@. 
For a width-$2$ ROABP, the map in \expref{Lemma}{lem:bivariatehs}
would be $(x_1,x_2) \mapsto (t^2 , t^2+t)$. 
However, $f(t^2, t^2+t) =0$ (over $\F_2$).
Hence, the hitting set does not work.
 
Now, we move on to getting a hitting set for an $n$-variate ROABP\@.

\subsection{\texorpdfstring{$n$}{n}-variate ROABP}
Observe that the map given in \expref{Lemma}{lem:bivariatehs} works 
irrespective of the degree of the polynomial,
as long as the field characteristic is large enough. 
We plan to obtain a hitting set for general $n$-variate ROABP
by applying this map recursively. 
For this, we use the standard divide and conquer technique. 
First, we make pairs of consecutive variables in the ROABP\@.
For each pair $(x_{2i-1}, x_{2i})$, we apply the map from \expref{Lemma}{lem:bivariatehs},
using a new variable $t_i$. 
Thus, we go to $n/2$ variables from $n$ variables. 
In \expref{Lemma}{lem:halving}, we use a hybrid argument to show that after this substitution
the polynomial remains nonzero. Moreover, the new polynomial
can be computed by a width-$w$ ROABP\@. 
Thus, we can again use the same map on pairs of new variables. 
By repeating the halving procedure $\log n$ times we get a univariate
polynomial.
In each round the degree of the polynomial gets multiplied by $w$. 
Hence, after $\log n$ rounds, the degree of the univariate polynomial
is bounded by $w^{\log n}$ times the original degree. 
Without loss of generality, let us assume that $n$ is a power of $2$.

\begin{lemma}[Halving the number of variables]
\label{lem:halving}
Suppose $\cha(\F)=0$, or $\cha(\F) >{} d$.
Let $f(\x)= D_1(x_1)D_2(x_2) \cdots D_n(x_n)$ be a nonzero polynomial 
of
individual 
degree~$d$
and
computed by a width-$w$ ROABP, where $D_1 \in \F^{1 \times w}[x_1]$,
$D_n \in \F^{w \times 1}[x_n]$ and $D_i \in \F^{w \times w}[x_i]$
for all $2 \leq i \leq n-1$. 
Let the map $\phi \colon \x \to \F[\t]$ be such that
for any index $1 \leq i \leq n/2$,
\begin{align*}
\phi(x_{2i-1}) &= t_i^w, \\  %
\phi(x_{2i}) &= t_i^w + t_i^{w-1}. %
\end{align*}
Then $f(\phi(\x)) \neq 0$.
Moreover, the polynomial $f(\phi(\x)) \in \F[t_1,t_2, \dots, t_{n/2}]$ 
is computed by a width-$w$ ROABP in the variable order $(t_1,t_2,\dots, t_{n/2})$.
\end{lemma}
\begin{proof}
Let us apply the map in $n/2$ rounds, \ie, 
define a sequence of polynomials $(f= f_0, f_1, \dots, f_{n/2}=f(\phi(\x)))$ 
such that the polynomial $f_i$ is obtained by replacing 
$(x_{2i-1},x_{2i})$ with $(\phi(x_{2i-1}), \phi(x_{2i}))$ in $f_{i-1}$
for each $1 \leq i \leq n/2$.
We will show that for each $1 \leq i \leq n/2$, 
if $f_{i-1} \neq 0$ then $f_i \neq 0$. 
Clearly this proves the first part of the lemma. 

Note that $f_{i-1}$ is a polynomial over variables $\{t_1,\dots, t_{i-1}, x_{2i-1}, \dots, x_n\}$. 
As $f_{i-1} \neq 0$, 
there exists a constant tuple $\alpha \in \F^{n-i-1}$
such that after replacing the variables $(t_1, \dots, t_{i-1},$ $ x_{2i+1}, \dots, x_{n})$
with $\alpha$, 
the polynomial %
$f_{i-1}$ remains nonzero.
After this replacement we get a polynomial $\hat{f}_{i-1}$ in the variables $(x_{2i-1},x_{2i})$.
As $f$ is computed by the ROABP $D_1D_2 \cdots D_n$,
the polynomial $\hat{f}_{i-1}$ can be written as $A^{\transpose} D_{2i-1}(x_{2i-1}) D_{2i}(x_{2i}) B$
for some $A, B \in \F^{w \times 1}$.
In other words, $\hat{f}_{i-1}$ has a bivariate ROABP of width-$w$. 
Thus, $\hat{f}_{i-1}(\phi(x_{2i-1}), \phi(x_{2i}) )$ is nonzero from \expref{Lemma}{lem:bivariatehs}. 
But, $\hat{f}_{i-1}(\phi(x_{2i-1}), \phi(x_{2i}) )$
is nothing but the polynomial obtained after 
replacing the variables $(t_1, \dots, t_{i-1}, x_{2i+1}, \dots, x_{n})$ in $f_i$
with $\alpha$. 
Thus, $f_i$ is nonzero. This finishes the proof.

Now, we argue that $f(\phi(\x))$ has a width-$w$ ROABP\@.
Let $\tilde{D}_i := D_{2i-1}(t_i^w) D_{2i}(t_i^w + t_i^{w-1})$
for all $1 \leq i \leq n/2$.
Clearly, $\tilde{D}_1 \tilde{D}_2 \cdots \tilde{D}_{n/2}$ is a width-$w$ ROABP computing $f(\phi(\x))$ in variable order
$(t_1,t_2, \dots, t_{n/2})$,
as $\tilde{D}_1 \in \F^{1 \times w}[t_1]$,
$\tilde{D}_{n/2} \in \F^{w \times 1}[t_{n/2}]$ and $\tilde{D}_i \in \F^{w \times w}[t_i]$
for all $2 \leq i \leq n/2-1$. 
\end{proof}

By applying the map $\phi$ in \expref{Lemma}{lem:halving},
we reduced an $n$-variate ROABP to an $(n/2)$-variate ROABP,
while preserving the non-zeroness. The resulting ROABP has
same width-$w$, but the individual degree goes up to become $2dw$, 
where $d$ is the original individual degree. 
As our map $\phi$ is degree insensitive, we can apply a similar map again
on the variables $\{t_i\}_{i=1}^{n/2}$.
That is, for $1\leq i \leq n/4$, define $\phi(t_{2i-1}) = s_i^{w}$ and $\phi(t_{2i}) = s_i^{w} + s_i^{w-1}$
for variables $\{s_1,s_2,\dots, s_{n/4}\}$.
Now, we get an $(n/4)$-variate ROABP 
of
individual degree $4dw^2$.
It is easy to see that when the map $\phi$ is repeatedly applied
in this way
$\log n$ times, we get a nonzero univariate polynomial of degree $ndw^{\log n}$.
Next lemma puts it formally. 
For ease of notation, we use the variable numbering from $0$ to $n-1$. 
Let $p_0(t) = t^w$ and $p_1(t)=t^w+t^{w-1}$.

\begin{lemma}
\label{lem:recursivehs}
Suppose $\cha(\F)=0$, or $\cha(\F) \geq ndw^{\log n}$.
Let $f \in \F[\x]$ be a nonzero polynomial 
of
individual degree $d$ and computed by a width-$w$
ROABP in variable order $(x_0,x_1, \dots, x_{n-1})$.
Let the map $\phi \colon \{x_0 ,x_1, \dots, x_{n-1}\} \to \F[t]$ be such that
for any index $0 \leq i \leq n-1$,
$$ \phi (x_i) = p_{i_1}( p_{i_2} \cdots (p_{i_{\log n}} (t) ) )\,,$$
where $i_{\log n} \; i_{\log n -1} \; \cdots \; i_{1}$ is the binary representation of $i$.

Then $f(\phi(\x))$ is a nonzero univariate polynomial 
of
degree $ndw^{\log n}$.
\end{lemma}
Note that the map $\phi$ crucially uses the knowledge of the variable order.	
In the last round when we are going from
two variables to one, the individual degree is $ndw^{\log n -1}$ and
 \expref{Lemma}{lem:bivariatehs} requires $\cha(\F)$ to be higher than the individual degree.
Thus, having $\cha(\F) \geq ndw^{\log n}$ suffices.
Hence, we get the following theorem.

\begin{theorem}
\label{thm:knownROABPhs}
Let $\C$ be the class of $n$-variate, 
individual-degree-$d$  %
polynomials computed by width-$w$ ROABPs.
Then a hitting set for $\C$ of size $O(ndw^{\log n})$ can be constructed,
when the variable order is known and the 
field characteristic is zero or at least $ndw^{\log n}$.
\end{theorem}
\begin{proof}
Let $f(\x)$ be a polynomial in class $\C$. 
{}From \expref{Lemma}{lem:recursivehs}, $f(\phi(\x)) \in \F[t]$ is a nonzero univariate polynomial 
of
degree
$ndw^{\log n}$.
Thus, if we substitute $1+ndw^{\log n}$ field values for the variable $t$,
one of them will keep $f(\phi(\x))$ nonzero. 
\end{proof}

{}From this, we immediately get the following result for constant-width ROABPs.
Note that when $w$ is constant, the lower bound on the characteristic also becomes $\poly(n)$.
\begin{corollary}
For the class of $n$-variate, 
individual-degree-$d$  
polynomials computed by
 constant width ROABPs (known variable order), 
a $\poly(n,d)$-size hitting set can be constructed,
when the field characteristic is zero (or larger than $poly(n,d)$).
\label{cor:constantwidth}
\end{corollary}

As mentioned earlier, our approach can potentially lead to a
polynomial-size 
hitting set for ROABPs. 
We make the following conjecture for which we hope to get a proof
on the lines of \expref{Lemma}{lem:bivariatehs}.

\begin{conjecture}
\label{con:polyROABP}
Suppose $\cha(\F)=0$.
Let $f(\x) \in \F[\x]$ be an $n$-variate, degree-$d$ polynomial computed by a width-$w$ ROABP\@.
Then $f(t^r,(t+1)^{r}, \dots, (t+n-1)^r) \neq 0$ for some $r$ bounded by $\poly(n,w,d)$.
\end{conjecture}

\section{Any-order ROABP} 
\label{sec:any-order-ROABP} %
In this section, we give better hitting sets for
the class of polynomials computed by ROABPs in any order. %
Recall that a polynomial $f(\x) \in \F[\x]$ is said to be computed by 
width-$w$ ROABPs in any order %
if for any permutation 
$\sigma \colon [n] \to [n]$, 
$f(\x)$ can be written as $$A^\transpose D_1(x_{\sigma(1)}) D_2(x_{\sigma(2)}) \cdots D_n(x_{\sigma(n)}) B$$
for some 
$D_i \in \F^{w \times w}[x_{\sigma(i)}]$ for $ 1 \leq i \leq n$ and 
$A, B \in \F^{w \times 1}$.

We will also consider ROABPs which compute a polynomial over the matrix algebra, that is,
polynomials whose coefficients are matrices. 
$D(\x) \in \F^{w \times w}[\x]$ is said to be computed by a width-$w$ ROABP
if $D(\x) =  D_1D_2 \cdots D_n $
for some polynomials $D_i \in \F^{w \times w}[x_{\sigma(i)}]$ for $1 \leq i \leq n$.
Forbes et al.~\cite{FSS14}
gave a hitting set of size $d^{O(\log w)}(nw)^{O(\log \log w)}$
for $n$-variate, individual-degree-$d$ polynomials computed by width-$w$ ROABPs in any order.
Note that when $d$ is small, this hitting-set size
is much better than that for 
ROABPs (with a particular variable order), \ie, $(ndw)^{O(\log n)}$ \cite{AGKS15}.
However when $d$ is $\Omega(n)$, the size is comparable to the latter case.
We improve the hitting-set size for 
polynomials computed by ROABPs in any order
 to $(ndw)^{O(\log \log w)}$. 
This is significantly better than the 
case of polynomials computed by ROABPs in a particular variable order.

\subsection{Rank-concentration}
Forbes et al.~\cite{FSS14} constructed the hitting set using the notion of
rank-concentration defined by Agrawal et al.~\cite{ASS13}.
Recall that $D(\x)$ is a polynomial over an algebra if its coefficients come from the algebra.

\begin{definition}[\cite{ASS13}]
A polynomial $D(\x)$ over an algebra is said to be $\ell$-concentrated
if its coefficients of $(<\ell)$-support monomials span all its coefficients.
That is, for all $\a \in \Wh^n$
\begin{equation}  \coeff_{D}(\x^{\a}) \in \Span\{\coeff_{D}(\xb) \mid \b \in \Wh^n, \; \supp(\b)< \ell \}\,.
\label{eq:lowspan}
\end{equation}
\end{definition}

Note that for a nonzero polynomial over a field, $\ell$-concentration simply means that 
one of its monomials of support $<\ell$ has a nonzero coefficient.
As we will see later, it is easy to
construct hitting sets for a polynomial which has low-support concentration. 
However, not every polynomial has a low-support concentration, for example, consider the following polynomial over a field:
$f(\x) = x_1x_2 \ldots x_n$.
Agrawal et al.~\cite{ASS13} observed that concentration can be achieved by a shift of variables,
\eg, $f(\x + \boldsymbol{1}) = (x_1+1)(x_2+1) \cdots (x_n+1)$ has $1$-concentration.
For a polynomial $f(\x)$, shift by a tuple $\s = (s_1,s_2, \dots, s_n)$
would mean $f(\x +\s) = f(x_1+s_1,x_2+s_2, \dots, x_n+s_n)$.

To achieve concentration, it is often useful to consider shifts which are polynomials. 
In particular, we will be considering shifts by bivariate polynomials,
\ie,  $\s(t_1,t_2) \in \F[t_1,t_2]^n$.
As ultimately we are interested in hitting sets, 
the variables $t_1$ and $t_2$ can later be replaced by field values. 
The size of the hitting set, in this case, will be multiplied by $\delta^2$, 
where $\delta$ is the maximum degree of any $s_i(t_1,t_2)$.
Thus, for a bivariate shift $\s(t_1,t_2)$, its degree will be viewed as the complexity measure. 
Note that for a polynomial $D(\x) \in \F^{w \times w}[\x]$, 
the coefficient of a monomial $\x^\a$ in $D(\x+\s(t_1,t_2))$ will be from $\F[t_1,t_2]^{w \times w}$.
So, when we talk of low-support concentration in $D(\x+\s(t_1,t_2))$, 
the span in~\eqref{eq:lowspan} is taken over the field $\F(t_1,t_2)$.

Forbes et al.~\cite{FSS14} construct the hitting set for 
polynomials computed by ROABPs in any order %
in two steps.
Let $f(\x)$ be an $n$-variate 
individual-degree-$d$  %
polynomial computed by width-$w$ 
ROABPs in any order. %
Their first step is to construct a tuple $\s(t_1,t_2)$ of bivariate polynomials
of
degree $\poly(n) d^{O(\log w)}$
such that $f(\x+\s)$ has
$O(\log w)$-concentration. 
We improve this step by constructing a new tuple $\s(t_1,t_2)$
of
degree
$(ndw)^{O(\log \log w)}$, which has the same property.

We follow the second step of Forbes et al.~\cite{FSS14} as it is. 
It is easy to see that $f(\x+\s)$ can also be computed by width-$w$ 
ROABPs in any order %
(over the field $\F(t_1,t_2)$).
They show that
if a given 
polynomial, computed by ROABPs in any order,
is $\ell$-concentrated
then there is a hitting set for it of size $(ndw)^{O(\log \ell)}$.
This implies a hitting set $\Hit$ of size $(ndw)^{O(\log \log w)}$ for $f(\x +\s)$.
Clearly, the set $\{\h + \s \mid \h \in \Hit \}$ is a hitting set for $f(\x)$.
One can obtain a hitting set in $\F^n$ by replacing $t_1$ and $t_2$ with sufficiently many
field values. 
By Schwartz-Zippel-DeMillo-Lipton Lemma, it will suffice to take more than $\deg_{t_1,t_2}(f(\h+\s)) = \deg(f) \cdot \deg(\s)$ values.
Thus, the final hitting-set size becomes 
$ \deg(\s) \cdot (ndw)^{O(\log \log w)}$.
With our improved bound on $\deg(\s)$, we get a hitting set of the desired size.

Now, we elaborate the first step of Forbes et al.~\cite{FSS14}, \ie, the construction of 
the shift $\s(t_1,t_2)$.
To achieve concentration they use the idea of Agrawal, Saha and Saxena~\cite{ASS13},
\ie, achieving concentration in small sub-ROABPs implies concentration in the
given ROABP\@.
For the sake of completeness, we rewrite the lemma using the terminology of this paper.
We first clarify a notation which will be used often:
for an $n$-tuple $\s$ and a polynomial $D(\x)$ which only depends variables 
$(x_{i_1},x_{i_2}, \dots, x_{i_\ell})$,
the expression %
 $D(\x+\s)$ will denote 
$D(x_{i_1}+s_{i_1},x_{i_2}+s_{i_2}, \dots, x_{i_\ell}+s_{i_\ell})$.

\begin{lemma}[\cite{ASS13,FSS14}]
\label{lem:ellton}
Let $\ell< n$ be any number.
Let $\s $ be the $n$-tuple  such that for any 
distinct $ i_1,i_2,\dots, i_\ell \in [n]$ and
individual-degree-$d$ %
polynomial
\[
  D(\x) = D_1(x_{i_1})D_2(x_{i_2}) \cdots D_\ell(x_{i_\ell})
\]
over the matrix algebra $\F^{w \times w}$, 
$D(\x+\s)$ is $\ell$-concentrated. 
Then for any 
individual-degree-$d$  %
polynomial $f(\x) \in \F[\x]$ computed by width-$w$ 
ROABPs in any order, %
$f(\x +\s)$ is $\ell$-concentrated.
\end{lemma}
\begin{proof}
Let $f'(\x) = f(\x+\s)$.
Consider any monomial $\xa$ with support $\geq \ell$. We will show
that its coefficient in $f'(\x)$ is in the span of smaller support coefficients 
in $f'(\x)$. 
Let $S=\{x_{i_1}, x_{i_2}, \dots, x_{i_\ell}\}$ be a set of $\ell$ variables contained in 
the support of monomial $\xa$. 
Let $\overline{S} = \{x_{i_{\ell+1}},\dots, x_{i_n} \}$ be the rest of the variables.
Let us write $\xa = \xb \xc$ with $\Supp(\b) = S$ and $\Supp(\c) \subseteq \overline{S}$.
Since, $f(x)$ is computed by 
ROABPs in any order, %
it has an ROABP in the variable order
$(x_{i_1}, \dots, x_{i_\ell},x_{i_{\ell+1}},\dots, x_{i_n})$. 
That is, 
$$f(\x) = A^\transpose D_1(x_{i_1})\cdots D_\ell(x_{i_\ell}) D_{\ell+1}(x_{i_{\ell+1}}) \dots D_n(x_{i_n}) B$$
for some $D_j \in \F^{w\times w}[x_{i_j}]$ for $1 \leq j \leq n$ and $A, B \in \F^{w \times 1}$. 
Let
\[
  D(\x) :=  D_1(x_{i_1})\cdots D_\ell(x_{i_\ell})\qquad\text{and}\qquad E(\x) := D_{\ell+1}(x_{i_{\ell+1}}) \dots D_n(x_{i_n})\,.
\]
Let $D'(\x) = D(\x+\s)$ and $E'(\x) = E(\x+\s)$. 
Clearly $f'(\x) = A^\transpose D'(\x)E'(\x) B$. 
By the lemma hypothesis, $D'(\x)$ is $\ell$-concentrated.
That is, 
\begin{equation}
\label{eq:Eell}
\coeff_{D'}(\xb) \in \Span\{\coeff_{D'}(\x^{\b'}) \mid \Supp(\b') \subseteq S, \; \supp(\b')< \ell \}\,.
\end{equation}
Note that we have $\Supp(\b') \subseteq S$ because each monomial in $D'(\x)$ comes from set $S$.
It is easy to see that for any monomial $\x^{\b'}$ with $\Supp(\b') \subseteq S$
$$\coeff_{f'}(\x^{\b'} \xc) = A^\transpose \coeff_{D'}(\x^{\b'}) \coeff_{E'} (\xc) B\,.$$
Thus, by left multiplying $A^\transpose$ and right multiplying $\coeff_{D'} (\xc) B$
 in~\eqref{eq:Eell}, we get
\[
\coeff_{f'}(\xa) \in \Span\{\coeff_{f'}(\x^{\b'}\xc) \mid \Supp(\b') \subseteq S, \; \supp(\b')< \ell \}\,.
\]
Note that $\supp(\b')+\supp(\c) < \supp(\b) + \supp(\c) = \supp(\a)$.
So, we can write
\[
\coeff_{f'}(\xa) \in \Span\{\coeff_{f'}(\x^{\a'}) \mid  \supp(\a')< \supp(\a) \}\,.
\]
In other words, for any monomial $\xa$ with $\supp(\a)\geq \ell$, 
$\coeff_{f'}(\xa)$ is in the span of coefficients of support smaller than $\supp(\a)$.
This would mean that, in fact, all coefficients of $f'(\x)$ are in the span of coefficients
with support $< \ell$.
\end{proof}

Now, for some $\ell \leq n$, the goal is to construct an $n$-tuple $\s$
such that for any distinct $ i_1,i_2,\dots, i_\ell \in n $, shifting by $\s$
ensures $\ell$-concentration in any $\ell$-variate ROABP of the form
$D(\x)=D_{1} (x_{i_1}) D_{2} (x_{i_2}) \cdots D_{\ell} (x_{i_\ell})$. 
Note that \expref{Lemma}{lem:ellton} holds for any value of $\ell \leq n$. 
However, one cannot choose $\ell$ to be arbitrary small.
The reason is that for an $\ell$-variate polynomial over a $k$-dimensional algebra,
one can hope to achieve $\ell$-concentration only when $\ell \geq \log (k+1)$.
To see this, consider the polynomial $D(\x) = \prod_{i=1}^{\ell} (1 + v_i x_i)$ 
over the algebra of $k\times k$ diagonal matrices, with $k = 2^{\ell}$. 
Here, $1$ stands for the matrix $\diag(1,1,\dots,1)$.
Define $v_1 = \diag(\alpha_1,\alpha_2,\dots, \alpha_k)$ for some distinct $\alpha_i$s. 
And define $v_i = v_1^{2^{i-1}}$ for $2 \leq i \leq \ell$.
It is not hard to see that the $2^{\ell}$ coefficients 
of the polynomial $D$ are $\{1,v_1,v_1^2,\dots, v_1^{2^{\ell}-1}\}$,
which are linearly independent. 
Note that since shifting is an invertible operation, the $2^\ell$ coefficients of 
$D(\x+\s)$ will also be linearly independent for any $\s$.
But, there are only $2^\ell -1$ monomials with support $< \ell$.
Hence, the coefficients of $(<\ell)$-support monomials 
cannot span all the coefficients in $D(\x+\s)$, for any shift $\s$.

Note that the dimension of the algebra $\F^{w \times w}$ is bounded by $w^2$.
To reiterate the goal, given $n$ and $w$,
we fix $\ell = \ceil{\log(w^2+1)}$ and we want to achieve $\ell$-concentration in 
all polynomials computed by an ROABP of the form $D_{1}  D_{2}  \cdots D_{\ell} $
where $D_j \in \F^{w \times w}[x_{i_j}]$ for $1 \leq j \leq \ell$,  
for some distinct $i_1,i_2,\dots,i_\ell \in [n]$.
As now we are dealing with polynomials in a small number of variables, 
it should be easier to achieve the concentration. 

Towards this goal, Forbes et al.~\cite{FSS14} give a bit more general result.
For any $\ell \geq \log(w^2+1)$, 
they construct a tuple $\s \in \F[t_1,t_2]^n$ of degree $\poly(n) d^{O(\ell)}$
which has the following property:
for any polynomial $D(\x) \in \F^{w \times w}[\x]$ which uses
at most $\ell$ of the $n$ variables and has 
individual-degree bound $d$, %
$D(\x+\s)$ has $\ell$-concentration. 
Here, Forbes et al.~\cite{FSS14} do not need that $D(\x)$ is computed by an ROABP\@.

We, on the other hand, use the property that $D(\x)$ is computed by a width-$w$,
 $\ell$-variate ROABP
and reduce
the degree of $\s(t_1,t_2)$ to $(n d w)^{O(\log \ell)}$.
Our construction of $\s(t_1,t_2)$ comes from the basis isolating weight assignment for ROABPs
from Agrawal et al.~\cite{AGKS15}.
We use the fact that for any polynomial over a $k$-dimensional algebra,
 shift by a basis isolating map achieves $\log(k+1)$-concentration~\cite{GKST15}.
\subsection{Basis isolation}
Let us first recall the definition of a basis isolating weight assignment. 
Let $\M$ denote the set of all monomials over the variable set ${\x}$
with 
individual-degree  %
$\leq d$.
Any function $\w \colon \x \to \Wh$ can be naturally extended to the set of all monomials 
as follows: $\w(\prod_{i=1}^n x_i^{\gamma_i}) = \sum_{i=1}^n \gamma_i \w(x_i)$, 
for any $(\gamma_i)_{i=1}^n \in \Wh^n$.
Note that if the variable $x_i$ is replaced with $t^{\w(x_i)}$ for each $i$,
then any monomial $m$ just becomes $t^{\w(m)}$.
Let $\A_k$ denote a $k$-dimensional algebra.

\begin{definition}[\cite{AGKS15}]
A weight function $\w \colon \x \to \Wh$ is called a basis isolating
weight assignment for a polynomial $D(\x) \in \A_k[\x]$, 
if there exists 
a set of monomials $S \subseteq \M$ ($ \abs{S} \leq k$)
whose coefficients form a basis for the coefficient space of $D({\x})$, 
such that 
\begin{itemize}
\item[--] for any $m, m' \in S$, $\w(m) \neq \w(m')$ and
\item[--] for any monomial $m \in \M \setminus S$, 
$$ \coeff_D(m) \in \Span \{ \coeff_D(m') \mid m' \in S, \; \w(m') < \w(m) \}\,.$$
\end{itemize}
\end{definition}

Gurjar et al.~\cite[Lemma 5.2]{GKST15} have shown that shifting by 
a basis isolating weight assignment achieves concentration.
We write their lemma here without a proof. 
For a weight function $\w \colon \x \to \Wh$, let $t^\w$ denote the 
tuple $(t^{\w(x_1)}, t^{\w(x_2)}, \dots, t^{\w(x_n)})$.

\begin{lemma}[Isolation to concentration]
Let $D(\x)$ be a polynomial over a $k$-dimensional algebra.
Let~$\w$ be a basis isolating weight assignment for~$D(\x)$.
Then $D(\x + t^\w)$ is $\ell$-concentrated (over $\F(t)$), where $\ell = \ceil{\log(k+1)}$.
\label{lem:lconc}
\end{lemma}

We now recall the construction complexity of a basis isolating weight assignment for
ROABP from~\cite{AGKS15}. Here, we present a slightly modified version of their Lemma~8 (without proof),
which easily follows from it.

\begin{lemma}
\label{lem:ellweightFunction}
For any numbers $\ell$, $n$, $k$ and $d$, we can construct a family $\W$
 of $(k n d)^{O(\log \ell)}$ integer weight assignments on variables $\{x_1,x_2,\dots,x_n\}$
with weights bounded by $(k n d)^{O(\log \ell)}$ which has the following property:
Let $D({\x})$ be an 
individual-degree-$d$   %
polynomial over $\A_k$ 
of the form $D_1(x_{i_1}) D_2(x_{i_2}) \dotsm D_{\ell}(x_{i_{\ell}})$
for some distinct $i_1,i_2,\dots,i_\ell \in [n]$.
Then one of the weight assignments in $\W$ is basis isolating for $D(\x)$.
\end{lemma}

Let $\W$ be the family constructed in \expref{Lemma}{lem:ellweightFunction} 
with $k = w^2$ and $\ell = \ceil{\log(w^2+1)}$.
{}From \expref{Lemma}{lem:ellweightFunction} and \expref{Lemma}{lem:lconc}, 
for any $D(\x)=D_{1} (x_{i_1}) D_{2} (x_{i_2}) \cdots D_{\ell} (x_{i_\ell}) \in \F^{w \times w}[\x]$
there exists a weight assignment $\w \in \W$ such that 
$D(\x+t^\w)$ is $\ell$-concentrated (over $\F(t)$).
However, we want a single tuple $\s$ which works for every $D(\x)$. 
To get a single tuple, we combine the tuples in $\{t^\w\}_{\w \in \W}$ 
using the standard technique of Lagrange interpolation (also used in~\cite{FSS14,GKST15}).
Let $\{\alpha_\w\}_{\w \in \W}$ be distinct constants.
Define 
\[ \s(t_1,t_2) = \sum_{\w \in \W } t_1^\w \prod_{\substack{{\w' \in \W}\\  
{\w' \ne \w}}} \frac{t_2-\alpha_{\w'}}{\alpha_{\w}-\alpha_{\w'}}\,.
\]
Note that $\s(t_1,\alpha_\w) = t_1^{\w}$.
The following claim shows that 
if $D(\x+t_1^\w)$ is $\ell$-concentrated for some $\w \in \W$, 
then $D(\x+\s(t_1,t_2))$ is also $\ell$-concentrated.
\begin{claim}
For a polynomial $D(\x)$ over an algebra and a constant $\alpha_\w$, 
if $D'(\x) = D(\x+\s(t_1,\alpha_\w))$ has $\ell$-concentration (over $\F(t_1)$)
then so does $D''(\x) = D(\x+\s(t_1,t_2))$ (over $\F(t_1,t_2)$). 
\label{cla:interpolation}
\end{claim}
\begin{proof}
It is easy to see that for any tuple $\s$, 
coefficients of $D(\x+\s)$ are linear combinations 
of coefficients of $D$ and vice versa (over an appropriate field). 
And since shifting is an invertible, it preserves the rank of all coefficients. 
That is,
$$\rank_\F\{\coeff_D(\x^\a)\}_{ \xa \in \M} 
= \rank_{\F(t_1)}\{\coeff_{D'}(\x^\a)\}_{ \xa \in \M}
=\rank_{\F(t_1,t_2)}\{\coeff_{D''}(\x^\a)\}_{ \xa \in \M}\,.
$$
Let this rank be $k$. Let us represent each coefficient of $D$ as a vector in $\F^k$.
Then coefficients of $D'$ and $D''$ come from $\F[t_1]^k$ and $\F[t_1,t_2]^k$, respectively.
Let $\M_{\ell} = \{ \xa \in \M \mid \supp(\a) < \ell \}$.
Since $D'$ has $\ell$-concentration,
$$\rank_{\F(t_1)}\left\{
\coeff_{D'}(\x^\a) \mid \xa \in \M_{\ell}
\right\} = k\,.$$
Hence, one can form an full rank matrix $L(t_1) \in \F[t_1]^{k \times k}$
which is given by
$$L(t_1) = \begin{pmatrix}
\coeff_{D'}(\x^{\a_1}) & \coeff_{D'}(\x^{\a_2}) &\dots & \coeff_{D'}(\x^{\a_k})
\end{pmatrix}
$$ 
for some $\x^{\a_1},\x^{\a_2},\dots,\x^{\a_k} \in \M_{\ell}$.
Define $L'(t_1,t_2) \in \F[t_1,t_2]^{k \times k}$ to be the matrix 
$$L'(t_1,t_2) = \begin{pmatrix}
\coeff_{D''}(\x^{\a_1}) &\coeff_{D''}(\x^{\a_2}) &\dots& \coeff_{D''}(\x^{\a_k})
\end{pmatrix}\,.
$$
{}From the definition of $D'$ and $D''$, 
it is clear that $L'(t_1,\alpha_\w) = L(t_1)$. 
Since $\det(L) \neq 0$, we get that $\det(L') \neq 0$.
Thus, 
$$\rank_{\F(t_1,t_2)}\left\{
\coeff_{D''}(\x^\a) \mid \xa \in \M_{\ell}
\right\} \geq k\,.$$
However, $k$ is the rank of all coefficients of $D''$.
Hence, $D''$ has $\ell$-concentration.
\end{proof}

Now, since $\s(t_1,t_2)$ has the desired property from \expref{Lemma}{lem:ellton},
$f(\x+\s(t_1,t_2))$ is $\ell$-concentrated for any polynomial $f(\x)$ computed by 
a width-$w$ ROABP\@.
Recall that $\deg_{t_1}(\s)$ is bounded by $( n d w)^{O(\log \log w)}$ from 
the construction in \expref{Lemma}{lem:ellweightFunction}. 
The same bound also holds on $\deg_{t_2}(\s)$ because $\abs{\W} = (n d w)^{O(\log \log w)}$.

\begin{lemma}
Given $n,d,w$, one can compute a tuple $\s(t_1,t_2) \in \F[t_1,t_2]^n$
of
degree $(ndw)^{O(\log \log w)}$ such that
for any $n$-variate, 
individual-degree-$d$
 polynomial $f(\x) \in \F[\x]$ computed by 
width-$w$ ROABPs in any order, %
$f(\x + \s(t_1,t_2))$ is $O(\log w)$-concentrated. 
\label{lem:ROABPconc}
\end{lemma}

As mentioned before, $O(\log w)$-concentration in $f(\x+\s)$ means that it has 
an $O(\log w)$-support monomial with a nonzero coefficient.
\expref{Lemma}{lem:ROABPconc} gives a bivariate tuple $\s(t_1,t_2)$ for the shift. 
We argue that one can substitute field values for $t_1$ and $t_2$ 
such that any chosen nonzero coefficient
in $f(\x+\s)$ remains nonzero after the substitution.
Note that any coefficient of $f(\x+\s)$ is a polynomial in $t_1$ and $t_2$ 
with its degree being at most $\deg(f) \cdot \deg(\s)$,
which is $(ndw)^{O(\log \log w)}$.
Thus, by Schwartz-Zippel-DeMillo-Lipton Lemma,
 substituting $(ndw)^{O(\log \log w)}$ many field values for $t_1$ and $t_2$ suffices. 

Now, we move on to the second step of Forbes, Shpilka and Saptharishi~\cite{FSS14}.
They give an $(ndw)^{O(\log \log w)}$-size hitting set for an already $O(\log w)$-concentrated 
polynomial which is computed by ROABPs in any order. %
They do this by reducing the PIT question to an $O(\log w)$-variate ROABP~\cite[Lemma 7.6]{FSS14}.

\begin{lemma}[\cite{FSS14}]
Let $f(\x) \in \F[\x]$ be an $n$-variate, 
individual-degree-$d$ 
polynomial 
computed by 
width-$w$ ROABPs in any order. %
Suppose $f(\x)$ has an $(\leq \ell)$-support monomial with a nonzero coefficient. 
Then, there is a $\poly(n,w,d)$-time computable $m$-variate 
map $\phi \colon \x \to \F[y_1,y_2, \dots, y_m]$ such that
$f(\phi(\x))$ is a nonzero polynomial 
of
degree $< d^2n^4$, where
$m = O(\ell^2)$.
Moreover, $f(\phi(\x))$ is computed by 
width-$w$, $m$-variate ROABPs in any order. %
\label{lem:fss}
\end{lemma}

{}From the results of~\cite{FS13,AGKS15}, we know that an $m$-variate, width-$w$ 
ROABP has an $(mdw)^{O(\log m)}$-size
hitting set. 
Combining \expref{Lemma}{lem:ROABPconc} and \expref{Lemma}{lem:fss} with this fact
and putting $m = O(\log^2 w)$,
we get the following.

\begin{theorem}
\label{thm:commROABPhs}
For the class of $n$-variate, 
individual-degree-$\degree$ 
polynomials computed by 
width-$w$ 
ROABPs in any order,
one can construct a hitting set of size $(n \degree w)^{O(\log \log w)}$.
\end{theorem}

\paragraph*{Concentration in Set-multilinear Circuits.}
Similar to \expref{Theorem}{thm:commROABPhs}, it would be interesting
to achieve the same size hitting set for set-multilinear circuits. 
Recall from \expref{Section}{sec:setmulti} that a polynomial computed by a
depth-$3$ set-multilinear circuit can be written as $(1,1, \dots, 1) \cdot D$,
where $D = D_1(\x_1) D_2(\x_2) \cdots D_q(\x_q)$ is a product of linear polynomials
over a commutative algebra of dimension $k$. 
Here the partition $\x = \x_1 \cup \x_2 \cup \cdots \cup \x_q$ is unknown.
Note that the polynomial $D$ can also be expressed as 
\[
  D = D_{\sigma(1)}(\x_{\sigma(1)}) D_{\sigma(2)}(\x_{\sigma(2)}) \cdots D_{\sigma(q)}(\x_{\sigma(q)})
\]
for any permutation $\sigma$ on $[q]$.
Hence, one can follow the same arguments as for 
polynomials computed by ROABPs in any order %
to get concentration in set-multilinear circuits.
Hence, we get the following result analogous to \expref{Lemma}{lem:ROABPconc}.

\begin{corollary}
Given $n,k$, one can compute an $n$-tuple $\s(t_1,t_2)$
of
degree $(nk)^{O(\log \log k)}$ such that
for any $n$-variate polynomial $f(\x)$ computed by a depth-$3$
set-multilinear circuit with top fan-in $k$,
the polynomial %
$f(\x + \s(t_1,t_2))$ is $O(\log k)$-concentrated. 
\end{corollary}
However, it is not clear whether the second step of the hitting-set construction
 can be done for set-multilinear circuits,
\ie, finding a better hitting set by assuming that the polynomial is already concentrated
(\expref{Lemma}{lem:fss}).

\section{Discussion}
For our first result (\expref{Theorem}{thm:knownROABPhs}), there are three directions for improvement.
Ideally, one would like to have all three at once.
\begin{enumerate}
\item Find a similar hitting set for the unknown-order case. 
In fact, we conjecture that the same hitting set (\expref{Lemma}{lem:recursivehs}) works
for the unknown-order case as well.
\item Get a hitting set for all fields (including low-characteristic fields). 
It is easy to construct examples over small characteristic fields where our hitting set does not work.
\item Reduce the hitting-set size to polynomial. 
To achieve this, it seems one has to do away with the divide and conquer approach.
\end{enumerate}
The map described in
\expref{Conjecture}{con:polyROABP} is a possible candidate for a 
polynomial-size 
hitting set for ROABPs
 and proving this conjecture would resolve two of the points above.

As mentioned earlier, we believe the ideas here may help in finding a better PRG for ROBPs. 
Studying such connections would in
particular take us closer towards resolving a major open question of
finding an $O(\log n)$-seed-length PRG for constant width ROBPs.
\section{Acknowledgements}
We thank the anonymous reviewer for suggesting that our techniques 
in \expref{Section}{sec:any-order-ROABP} %
might work 
for a more general
class of polynomials, namely polynomials computed by ROABPs in any order. %
 We are thankful to Herv\'{e} Fournier, Sumanta Ghosh, Ramprasad Saptharishi
for helpful discussions on the same.

\bibliographystyle{tocplain}   %
\bibliography{v013a002}

\begin{tocauthors}
\begin{tocinfo}[gurjar]
 Rohit Gurjar\\
 Postdoctoral fellow\\
 Tel Aviv University\\
 Tel Aviv, Israel\\  %
 rohitgurjar0\tocat{}gmail\tocdot{}com \\
 \url{http://www.cse.iitk.ac.in/users/rgurjar}
\end{tocinfo}
\begin{tocinfo}[korwar]
 Arpita Korwar\\
 \phd\ student\\  %
 Indian Institute of Technology Kanpur\\
 Kanpur, India\\ 
 arpk\tocat{}cse\tocdot{}iitk\tocdot{}ac\tocdot{}in\\
 \url{http://www.cse.iitk.ac.in/users/arpk}
\end{tocinfo}

\begin{tocinfo}[saxena]
  Nitin Saxena\\
  Associate professor\\
  Indian Institute of Technology Kanpur\\
  Kanpur, India\\ %
 nitin\tocat{}cse\tocdot{}iitk\tocdot{}ac\tocdot{}in\\
 \url{http://www.cse.iitk.ac.in/users/nitin}
\end{tocinfo}
\end{tocauthors}

\begin{tocaboutauthors}
\begin{tocabout}[gurjar]
\textsc{Rohit Gurjar} is currently a postdoctoral fellow at \href{https://english.tau.ac.il/}{Tel Aviv University}, working with \href{http://www.cs.tau.ac.il/~shpilka/}{Amir Shpilka}. 
His previous postdoc was at \href{http://www.htw-aalen.de/}{Aalen University} (2015-16) with \href{http://image.informatik.htw-aalen.de/Thierauf/}{Thomas Thierauf}. 
Before that he was at \href{http://www.iitk.ac.in}{Indian Institute of Technology Kanpur} for 10 long years for his B.\,Tech.-M.\,Tech.\ (2005-10) and
\phd\ (2010-15). 
He was very fortunate
to have \href{https://sites.google.com/view/manindra/home}{Manindra Agrawal} and \href{http://www.cse.iitk.ac.in/users/nitin}{Nitin Saxena} as his 
\phd\ supervisors. His \phd\ thesis was chosen for the ACM India Doctoral Dissertation Award, 2017.
He is in general interested in theoretical computer science, and in particular in computational complexity and derandomization.
Some problems on which he has worked
are polynomial identity testing, perfect matching, and matrix completion.
He likes hiking, cycling and listening to music.
\end{tocabout}

\begin{tocabout}[korwar]
\textsc{Arpita Korwar} is a \phd\ student at
\href{http://www.iitk.ac.in}{IIT Kanpur} advised by
\href{https://sites.google.com/view/manindra/home}{Manindra Agrawal} and
\href{http://www.cse.iitk.ac.in/users/nitin}{Nitin Saxena}.
Her work has also been influenced by
\href{https://www.cse.iitk.ac.in/users/sb/}{Somenath Biswas}.
Her interests are in the theoretical aspects of computer science,
especially in using algebraic techniques for computational complexity.
Currently she is working on polynomial identity testing and
arithmetic lower bounds with
\href{https://webusers.imj-prg.fr/~herve.fournier/}{Herv\'{e} Fournier} and
\href{https://webusers.imj-prg.fr/guillaume.malod}{Guillaume Malod} at
the %
\href{https://universite.univ-paris-diderot.fr/}{University of Paris
Diderot}.
She visited \href{http://image.informatik.htw-aalen.de/Thierauf/}{Thomas
Thierauf} at 
the %
\href{http://www.uni-ulm.de/en/}{University of Ulm}
a couple of times during her \phd\ studies.
She likes the outdoors and sports!
\end{tocabout}

\begin{tocabout}[saxena]
\textsc{Nitin Saxena} received his \phd\ from the \href{http://www.iitk.ac.in}{Indian Institute of Technology Kanpur} in 2006. He was truly fortunate to have \href{https://sites.google.com/view/manindra/home}{Manindra Agrawal} for his supervisor. He also spent stints 
at  %
\href{https://www.princeton.edu/main/}{Princeton University} (2003-04) and 
at the %
\href{http://www.nus.edu.sg/}{National University of Singapore} (2004-05). 
He was a postdoc 
at the   %
\href{https://www.cwi.nl/}{Centrum voor Wiskunde en Informatica Amsterdam} (2006-08) and a faculty 
at the  %
\href{http://www.hausdorff-center.uni-bonn.de/}{Hausdorff Center for Mathematics Bonn} (2008-13). Nitin's long-term interests are in 
algebra-flavored   %
computational complexity problems. 
He has contributed to primality testing, polynomial identity testing, polynomial independence, polynomial factoring and polynomial equivalence problems. Some of these works have been awarded 
the %
G\"odel prize, Fulkerson prize, CCC best paper (2006), 
and ICALP best paper (2011) %
awards.
He enjoys interacting with enthusiastic young 
researchers. 
In his spare (\& non-spare) time he enjoys listening to music, watching movies, reading non-fiction, swimming and traveling. 

\end{tocabout}
\end{tocaboutauthors}

\end{document}